\documentclass[a4paper,11pt]{article}
\usepackage{amsmath,amssymb,amsthm,fullpage}
\usepackage{authblk}
  \usepackage{paralist}
  \usepackage{graphics}
  \usepackage{epsfig} 
\usepackage{graphicx}  
\usepackage{epstopdf}

\newtheorem{theorem}{Theorem}[section]
\newtheorem{corollary}[theorem]{Corollary}

\newtheorem{lemma}[theorem]{Lemma}

\theoremstyle{definition}
\newtheorem{definition}[theorem]{Definition}

\newtheorem{example}[theorem]{Example}

\newcommand{\B}{\ensuremath{\mathcal{B}}} 
\newcommand{\zed}{\ensuremath{\mathbb{Z}}}

\renewcommand{\Pr}{\operatorname{Pr}}

\title{Splitting authentication codes with perfect secrecy: new results, constructions and connections with algebraic manipulation detection codes}
\author[1]{Maura B.~Paterson}
\author[2]{Douglas R.~Stinson}
\affil[1]{Department of Economics, Mathematics and Statistics, Birkbeck, University of London, Malet St, London WC1E 7HX, UK} 
\affil[2]{David R. Cheriton School of Computer Science, University of Waterloo, Waterloo, Ontario, N2L 3G1, Canada\footnote{D.R.~Stinson's research is supported by NSERC discovery grant RGPIN-03882}}
\begin{document}
\maketitle
\begin{abstract}
A splitting BIBD is  a type of combinatorial design that can be used to construct splitting authentication codes with good properties.  In this paper we show that a design-theoretic approach is useful in the analysis of more general splitting authentication codes.  Motivated by the study of algebraic manipulation detection (AMD) codes, we
 define the concept of a {\em group generated} splitting authentication code.  We show that all group-generated authentication codes have perfect secrecy, which allows us to demonstrate that algebraic manipulation detection codes can be considered to be a special case of an authentication code with perfect secrecy. 

We also investigate splitting BIBDs that can be ``equitably ordered''. These splitting BIBDs yield authentication codes with splitting that also have perfect secrecy.  We show that, while group generated BIBDs are inherently equitably ordered, the concept is applicable to more general splitting BIBDs. For various pairs $(k,c)$, we determine necessary and sufficient (or almost sufficient) conditions for the existence of $(v, k \times c,1)$-splitting BIBDs
that can be equitably ordered. The pairs for which we can solve this problem are $(k,c) = (3,2), (4,2), (3,3)$ and $(3,4)$, as well as all cases with $k = 2$. 
\end{abstract}

\section{Introduction}
The use of authentication codes for providing authentication in an unconditionally secure setting has long been studied, following models developed by Simmons \cite{Simmons84}.  Authentication codes with perfect secrecy ensure confidentiality of sources as well as authenticity.  There is a considerable literature on authentication and secrecy codes, including models that make different assumptions about the distribution of the sources \cite{Stinson88,Stinson1990}.  We observe that the majority of the focus has been on the case where, for a given key, there is a unique encoding for each source.  On the other hand, {\em Splitting authentication codes} allow multiple different encodings of a source under a specific key. Allowing splitting can facilitate better performance for certain parameter settings, and can also yield constructions that work for any source distribution.  There is also a wide literature on splitting authentication codes, including many constructions \cite{Blundo99,GMW,Huber12,Ogata2004,DeSoete91,Wang,WS}.  However, the case of splitting authentication codes with perfect secrecy has not been systematically considered.  

Our investigation of splitting authentication codes with perfect secrecy is motivated by consideration of the properties and structure of {\em algebraic manipulation detection (AMD) codes} with a view to better characterising those application contexts in which they can be usefully applied.  AMD codes were introduced by  Cramer, Dodis, Fehr, Padr\'{o} and Wichs in EUROCRYPT 2008 as a way of abstracting ideas used in the construction of robust secret sharing schemes into more general tools for providing robustness against active manipulation in cryptographic systems \cite{cramer08}.  The definitions of these objects have certain similarities with authentication codes, in that both aim to detect whether an adversary has tampered with an encoded element.  Connections noted in the literature include the use of AMD codes by Cramer et al.\ in the construction of a primitive they call a {\em KMS-MAC}, which could be viewed as a variant of an authentication code \cite{cramer08}.  However, there are also clear differences in the two definitions.  For example, authentication codes rely on the use of a shared key, whereas there are no keys involved in the definition of an AMD code.  Also, the underlying context for their use and the corresponding security definitions are different.  The definition of an authentication code is purely combinatorial, as are many of the known constructions, whereas an AMD code inherently requires the algebraic structure of an abelian group.  

In Section~\ref{sec:splittingsetsystem} of this paper we connect the combinatorial and algebraic perspectives by taking a design-theoretic approach to studying splitting authentication codes, with a particular focus on their automorphism groups.  We introduce the notion of a {\em group generated authentication code,} and show that the property of being group generated is sufficient to ensure the authentication code has perfect secrecy, and it also gives other desirable properties such as optimal protection against impersonation attacks.  We clarify the relationship between authentication codes and AMD codes by demonstrating that, in terms of their mathematical structure, an AMD code is a special case of a group generated authentication code, with weak AMD codes corresponding to authentication codes that require a uniform distribution on the sources and strong AMD codes corresponding to authentication codes that work for any source distribution.  We discuss the consequences of this connection for our understanding of AMD codes.

In Section~\ref{sec:equitableordering} we consider perfect secrecy for certain optimal authentication codes that are not necessarily group generated.  Splitting BIBDs are a type of combinatorial design that give rise to splitting authentication codes that are optimal with respect to certain bounds on the adversary's success probability in substitution attacks.  In Section~\ref{sec:equitableordering} we define the {\em equitable ordering} property for splitting BIBDs, which guarantees that the corresponding authentication codes offer perfect secrecy.  We give techniques to provide equitable ordering for splitting BIBDs with a range of parameters, which permits the conversion of a wide class of splitting authentication codes into splitting authentication codes with perfect secrecy.

\subsection{Authentication Codes}
An {\em authentication code} is a \makebox{4-tuple} $(\mathcal{S},\mathcal{T},\mathcal{K},\mathcal{E})$ where $\mathcal{ S}$ is a finite set of {\em sources}, the set $\mathcal{ T}$ is a finite set of {\em messages}, the set $\mathcal{ K}$ is a finite set of {\em keys} and $\mathcal{ E}$ is a set of {\em encoding rules}.    The encoding rules are (possibly randomised) maps from $\mathcal{S}$ to $\mathcal{T}$ that are indexed by the keys in $\mathcal{K}$.  We use the notation $e_k(s)\subset \mathcal{ T}$ to denote the set of possible encodings of source $s$ under the encoding rule $e_k$.  Note that, for distinct sources $s$ and $s^\prime$, we require $e_k(s)\cap e_k(s^\prime)=\emptyset$ for each $k\in\mathcal{ K}$; in practical terms, this means that knowledge of $k$ enables the unique identification of the source from the encoding.  We assume that the keys are drawn uniformly at random from $\mathcal{ K}$, independently of $s$.  A sender who shares a key $k\in \mathcal{ K}$ with a receiver authenticates a source value $s\in \mathcal{ S}$ by calculating a message $t\in e_k(s)$ and transmitting it to the receiver.  The receiver accepts the message as authentic if $t\in e_k(s)$.  
\begin{example}\label{ex:firstacode}
Let $\mathcal{ S}=\{0,1\}$ and $\mathcal{ K}=\{0,1,2,3,4\}$. We can define an authentication code by means of the following table.  To generate an encoding for source $s$ and key $k$, we choose one of the two entries in the corresponding row/column uniformly at random.
\begin{equation*}
\begin{array}{rcc}
\hline
k& e_k(0)&e_k(1)\\
\hline
0&\{1,4\}&\{2,3\}\\
1&\{2,0\}&\{3,4\}\\
2&\{3,1\}&\{4,0\}\\
3&\{4,2\}&\{0,1\}\\
4&\{0,3\}&\{1,2\}\\
\hline
\end{array}
\end{equation*}
If a receiver possesses the key $3$, for example, then they would accept the message $2$ as being an authentic encoding of the source $0$.  However, if they received the message $3$ they would reject this as being inauthentic.
\end{example}

If $|e_{k}(s)|=1$ for all $k\in\mathcal{ K}$, $s\in\mathcal{ S}$, then the authentication code is {\em deterministic}, otherwise it is said to be a {\em splitting authentication code}.  If $|e_{k}(s)|=c$ for all $k\in\mathcal{ K}$, $s\in\mathcal{ S}$, then we say the authentication code is {\em $c$-splitting}. In the case where the encoding of each source $s$ under any encoding rule $e_k$ is chosen uniformly from the messages in $e_k(s)$ the authentication code is said to have {\em equiprobable encoding}.  For instance, the authentication code described in Example~\ref{ex:firstacode} is a 2-splitting authentication code with equiprobable encoding.  For all authentication codes considered in this paper, we assume we have equiprobable encoding.  

There are several relevant probability distributions associated with an authentication code.  There is the distribution on the sources; in some circumstances we consider the case where this distribution is uniform, although we also consider authentication codes with arbitrary source distributions.  There is the distribution on the keys, which is generally assumed to be uniform and independent of the source distribution.  Additionally there is the distribution associated with encoding of a source $s$ under a key $k$, which we assume is uniform.  Finally, there is the resulting distribution induced on the space of messages.  For a key $k$, source $s$ and message $t$, the probability that the message $t$ results from encoding source $s$ under key $k$ can be expressed as
\begin{align*}
\Pr(k,s,t)&=\Pr(s)\Pr(k)\Pr(t|k,s),\\
&=\frac{\Pr(s)}{|\mathcal{ K}||e_k(s)|}. \intertext{For a $c$-splitting authentication code this becomes}
\Pr(k,s,t)&=\frac{\Pr(s)}{c|\mathcal{ K}|}.
\end{align*}

An adversary who has seen a valid message $t\in e_k(s)$ can try and trick the receiver into accepting as valid a different message $t^\prime$.  This attack is known as {\em substitution}, and it succeeds if $t^\prime\in e_k(s^\prime)$ for some source $s^\prime\neq s$.  It is desirable to construct authentication codes for which the  probability of a successful substitution attack is as small as possible.  We assume that the adversary is aware of the distribution from which the source is drawn, and in response they choose a {\em substitution strategy} $\sigma$ that consists of a choice of replacement message $\sigma(t)$ for each possible message $t\in\mathcal{ T}$.

Let $\sigma$ be a substitution strategy for attacking an authentication code $(\mathcal{ S},\mathcal{ T},\mathcal{ K},\mathcal{ E})$. If the key is $k\in\mathcal{ K}$ and the source is $s\in \mathcal{ S}$, then the  adversary's strategy succeeds whenever the message is a value $t$ from the set 
\begin{align*}
X^\sigma_{k,s}=\{t\in e_k(s): \sigma(t)\in e_k(s^\prime)\text{ for some }s^\prime\neq s\}.  
\end{align*}
The overall success probability $\epsilon_\sigma$ of the strategy $\sigma$ is given by
\begin{align}
\epsilon_\sigma&=\sum_{k\in \mathcal{ K}}\sum_{s\in \mathcal{ S}}\sum_{t\in X^\sigma_{k,s}}\Pr(k,s,t). \label{eq:subsuc}
\end{align}
The authentication code is said to have {\em substitution probability } at most $\epsilon$ if $e_\sigma\leq \epsilon$ for every strategy $\sigma$.
We observe that the expression in (\ref{eq:subsuc}) can be written as follows:
\begin{align}
\epsilon_\sigma=\sum_{k\in \mathcal{ K}}\sum_{s\in \mathcal{ S}}\frac{|X^\sigma_{k,s}|\Pr(s)}{|\mathcal{ K}||e_k(s)|}.\label{eq:gensub}
\end{align}
In \cite{Blundo99} it was shown that the substitution probability $\epsilon$ is at least
\begin{align}
\min_{k\in \mathcal{ K}} \frac{\left|\bigcup_{s\in\mathcal{ S}}e_k(s)\right|-\max_{s\in\mathcal{ S}}|e_k(s)|}{|\mathcal{ T}|-1}. \label{eq:optimalsub}
\end{align}
(This was a correction of a result from \cite{DeSoete91}.)  An authentication code for which this bound is satisfied is said to have {\em optimal substitution probability.}
\begin{example}
Consider the authentication code of Example~\ref{ex:firstacode}.  As this is a 2-splitting authentication code with 5 keys and equiprobable encoding, the success probability of a substitution strategy $\sigma$ is given by
\begin{align*}
\epsilon_\sigma&=\sum_{k\in \mathcal{ K}}\sum_{s\in \mathcal{ S}}\frac{|X^\sigma_{k,s}|\Pr(s)}{10}.
\end{align*}
We first observe that for any $t\in \mathcal{ T}$, if $\sigma(t)=t$ then it is the case that $t\notin X^\sigma_{k,s}$ for any choice of $k$ or $s$.  Consider now the element $0\in \mathcal{ T}$.  If $\sigma(0)=1$, then $0\in X^\sigma_{4,0}$ and $0\in X^\sigma_{2,1}$ but $0\notin X^\sigma_{k,s}$ for any other choice of $k$ and $s$.  Similarly, for any other nonzero choice of $\sigma(0)$, we can check that there is one value of $k$ with $0\in X^\sigma_{k,0}$ and one value of $k$ with $0\in X^\sigma_{k,1}$.  The same holds true for every other element $t$ of $\mathcal{ T}$: if $\sigma(t)\neq t$ then $t\in X^\sigma_{k,0}$ for precisely one value of $k$, and $t\in X^\sigma_{k,1}$ for precisely one value of $k$.  Thus for any strategy $\sigma$ it is the case that $\sum_{k\in\mathcal{ K}}|X^\sigma_{k,0}|\leq 5$, and also $\sum_{k\in\mathcal{ K}}|X^\sigma_{k,1}|\leq 5$.
Hence we have
\begin{align*}
\epsilon_\sigma&=\frac{1}{10}\sum_{s\in \mathcal{ S}}\sum_{k\in \mathcal{ K}}\Pr(s)|X^\sigma_{k,s}|,\\
&=\frac{1}{10}\sum_{k\in \mathcal{ K}}\left(\Pr(0)|X^\sigma_{k,0}|+\Pr(1)|X^\sigma_{k,1}|\right),\\
&\leq\frac{1}{10}\left(5\Pr(0)+5\Pr(1)\right),\\
&=\frac{1}{2}.
\end{align*}
Hence $\epsilon_\sigma\leq1/2$ for any $\sigma$, and we note further that $\epsilon_\sigma=1/2$ for any strategy $\sigma$ that satisfies $\sigma(t)\neq t$ for all $t\in \mathcal{ T}$.  This holds true for any source distribution.

If we consider (\ref{eq:optimalsub}) for this authentication code we have
\begin{align*}
\epsilon\geq \frac{4-2}{5-1}=\frac{1}{2},
\end{align*}
and so this authentication code has optimal substitution probability.
\end{example}


Another attack considered in the literature is that of {\em impersonation}, in which an adversary who has not seen any transmitted messages sends a message to the receiver in the hopes that it will be accepted as valid.  The probability that an adversary sending message $t$ succeeds is given by 
\begin{align}
\frac{\left|\{k\in\mathcal{ K}: t\in \bigcup_{s\in\mathcal{ S}}e_k(s)\}\right|}{|\mathcal{ K}|},\label{eq:imp1}
\end{align}
and the {\em impersonation probability} of the authentication code is the maximum over all $t$ of these success probabilities.  Simmons observed in \cite{Simmons84} that the impersonation probability of an authentication code is at least 
\begin{align} 
\min_{k\in \mathcal{ K}} \frac{\left|\bigcup_{s\in \mathcal{ S}}e_k(s)\right|}{|\mathcal{ T}|}.\label{eq:imp2}
\end{align}
An authentication code that meets this bound is said to have {\em optimal impersonation probability.}
\begin{example}
For the authentication code of Example~\ref{ex:firstacode} we observe that \begin{align*}|\{k\in \mathcal{ K}:t\in\bigcup_{s\in\mathcal{ S}}e_k(s)\}|=4\end{align*} for any choice of $t$, hence the impersonation probability is $4/5$.  This is in fact optimal, as $|\bigcup_{s\in \mathcal{ S}}e_k(s)|=4$, so the expression in (\ref{eq:imp2}) also evaluates to $4/5$.
\end{example}

\begin{definition} 
An authentication code $(\mathcal{ S},\mathcal{ T},\mathcal{ K},\mathcal{ E})$ has {\em perfect secrecy} if the message $t$ reveals no information about the source $s$, that is if
\begin{align*}
\Pr(s\mid t)=\Pr(s),
\end{align*}
for all $t\in \mathcal{ T}$ and $s\in \mathcal{ S}$.
\end{definition}

While generalisations of these notions where the adversary sees more than one message have been considered in the literature, e.g.\ \cite{Stinson1990}, in this paper we restrict our attention to the case where the adversary sees a single message.
\subsection{AMD codes}
An {\em algebraic manipulation detection code} (AMD code) is a \makebox{4-tuple} $(\mathcal{ S},\mathcal{ G},{ A},{ E})$, where $\mathcal{S}$ is a finite set of {\em sources}, $\mathcal{ G}$ is a finite additive group, ${A}\subset\mathcal{ G}$ is a set of {\em  valid encodings} and $E\colon \mathcal{ S}\rightarrow { A}$ is a (possibly randomised) encoding rule \cite{oldamd}.  We use the notation $A(s)\subset\mathcal{ G}$ to denote the set of valid encodings of source $s\in S$, and we require $A(s)\cap A(s^\prime)=\emptyset$ whenever $s\neq s^\prime$.  We have ${ A}=\cup_{s\in \mathcal{ S}}A(s)$, and we will often use the notation $\mathcal{ A}$ to denote the collection of disjoint subsets of $\mathcal{ G}$ given by $\{A(s): s\in\mathcal{ S}\}$.  We set $a_s=|A(s)|$ and $a=|A|=\sum_{s\in\mathcal{ S}}a_s$.  

A user selects a source $s\in \mathcal{ S}$ randomly according to a distribution that is known to the adversary then the encoding rule $E$ is used to encode $s$ as an element $g\in A(s)$.  If $g$ is chosen uniformly at random from $A(s)$, then the AMD code is said to have {\em equiprobable encoding}. Throughout this paper we assume all AMD codes we consider have equiprobable encoding.

\begin{example}\label{ex:firstamd}
Let $\mathcal{ S}=\{0,1\}$, let $\mathcal{ G}=\mathbb{Z}_9$, and let $\mathcal{ A}=\{\{0,1\},\{2,4\}\}$, so $A=\{0,1,2,4\}$.  We can construct an AMD code $(\mathcal{ S},\mathcal{ G},A,E)$ by defining an encoding rule $E$ that encodes the source $0$ as either $0$ or $1$, each with probability $1/2$, and encodes the source $1$ as either $2$ or $4$, each with probability $1/2$.  We typically refer to $\mathcal{ A}$ as an AMD code, since $E$ is implied once we assume equiprobable encodings.
\end{example}

An adversary selects an element $\Delta \in \mathcal{ G}$ to be added to $g$.  The user accepts $g+\Delta$ if it is a valid encoding of some source, that is, if $g+\Delta\in A(s^\prime)$ for some $s^\prime\in\mathcal{ S}$, in which case it is decoded to $s^\prime$.  The adversary wins if $s^\prime \neq s$, that is if their algebraic manipulation has succeeded in causing the user to decode the stored value incorrectly. Given a source $s\in \mathcal{ S}$ and an element $\Delta\in\mathcal{ G}$, define the set $X^\Delta_s$ to be 
\begin{align*}
X^\Delta_s=\{g\in A(s): g+\Delta\in A(s^\prime) \text{ for some }s^\prime\neq s\}.  
\end{align*}
Then the probability that an adversary who chooses $\Delta$ succeeds is
\begin{align*}
\epsilon_\Delta=\sum_{s\in \mathcal{ S}}\sum_{g\in X^\Delta_s} \Pr(s,g).
\end{align*}        

We observe that $\Pr(s,g)=\Pr(s)\Pr(g\mid s)$, and that this is equal to $\Pr(s)|A(s)|^{-1}$ as we have equiprobable encodings. This allows us to express $\epsilon_\Delta$ as
\begin{align*}
\epsilon_\Delta=\sum_{s\in \mathcal{ S}}\frac{|X_s^\Delta|\Pr(s)}{|A(s)|}.
\end{align*}

\begin{definition}
An AMD code with $|\mathcal{ S}|=m$ and $|\mathcal{ G}|=n$ is referred to as a {\em weak $(m,n,\epsilon)$-AMD code} if an adversary who does not know the source has success probability at most $\epsilon$ in the case where the sources are uniformly distributed.  Here we have
\begin{align*}
\epsilon_{\Delta}=\sum_{s\in \mathcal{ S}}\frac{|X_s^\Delta|}{m|A(s)|},
\end{align*}
and we require $\epsilon_\Delta\leq \epsilon$ for all $\Delta\in \mathcal{ G}^*$.
\end{definition}

\begin{example}
Consider the AMD code of Example~\ref{ex:firstamd}, and suppose an adversary chooses $\Delta=1$.  Then $X_0^1=\{1\}$ and $X_1^1=\emptyset$, so $\epsilon_1=\frac{1}{4}$.  Similar calculations show that in fact $\epsilon_\Delta=1/4$ for each $\Delta\in \mathbb{Z}
_9^*$, and so this is a weak $(2,9,1/4)$-AMD code.
\end{example}

\begin{definition}
An AMD code with $|\mathcal{ S}|=m$ and $|\mathcal{ G}|=n$ is a {\em strong $(m,n,\epsilon)$-AMD code} if the success probability of an adversary who knows the source is at most $\epsilon$.  Let $\epsilon_{s,\Delta}$ be the success probability of an adversary who selects the element $\Delta$, conditioned on the event that the source is $s$.  Then
\begin{align*}
\epsilon_{s,\Delta}=\frac{|X_s^\Delta|}{|A(s)|},
\end{align*}
and we require $\epsilon_{s,\Delta}\leq \epsilon$ for all $s\in \mathcal{ S}$ and $\Delta\in \mathcal{ G}^*$.
\end{definition}
\begin{example}
The AMD code of Example~\ref{ex:firstamd} has $\left|X_s^\Delta\right|\leq 1$ for each $s\in \mathcal{ S}$ and $\Delta\in \mathcal{ G}^*$, which implies $\epsilon_{s,\Delta}\leq 1/2$.  Thus it is a strong $(2,9,1/2)$-AMD code.
\end{example}
An $(m,n,\epsilon)$-AMD code is said to be {\em $c$-regular} if it has equiprobable encoding and $a_s=c$ for all $s\in\mathcal{ S}$.  A $1$-regular AMD code is {\em deterministic}. We observe that a deterministic AMD cannot be a strong $(m,n,\epsilon)$-AMD code for any $\epsilon<1$, since an adversary who knows the source and knows the encoding of the source (due to the fact the encoding is deterministic) has enough information to pick a value of $\Delta$ that will succeed.

We observe that, while it would also be possible to study AMD codes with a specified distribution on the sources that is not the uniform distribution, this has not been considered in the literature.  However, this notion does lead naturally to an alternative interpretation of strong AMD codes: rather than assume a model where the adversary knows the value of the source, we could instead view strong AMD codes as being ones that work for {\em any} source distribution, as demonstrated in the following theorem:
\begin{theorem}
An AMD code $(\mathcal{ S},\mathcal{ G},\mathcal{ A},{ E})$ is a {strong $(m,n,\epsilon)$-AMD code} if and only if the success probability of an adversary is at most $\epsilon$ for any choice of source distribution.
\end{theorem}
\begin{proof}
Suppose an adversary's success probability against an AMD code $(\mathcal{ S},\mathcal{ G},\mathcal{ A},{ E})$ is at most $\epsilon$ for any source distribution.  Then the adversary's success probability is at most $\epsilon$ for the distribution in which source $s$ is chosen with probability 1, for any $s\in S$.  Hence it is a {strong $(m,n,\epsilon)$-AMD code}.

Conversely, suppose $(\mathcal{ S},\mathcal{ G},\mathcal{ A},{ E})$ is a  {strong $(m,n,\epsilon)$-AMD code}.  Then for all $s\in\mathcal{ S}$ and $\Delta\in \mathcal{ G}$ we have
\begin{align*}
\frac{|X_s^\Delta|}{|A(s)|}\leq \epsilon.
\end{align*}
Let the sources be chosen according to a distribution that chooses source $s\in \mathcal{ S}$ with probability $\Pr(s)$.  Then for any $\Delta\in \mathcal{ G}$ we have
\begin{align*}
\epsilon_\Delta&=\sum_{s\in \mathcal{ S}}\frac{\Pr(s)|X_s^\Delta|}{|A(s)|},\\
&\leq \sum_{s\in \mathcal{ S}}\Pr(s)\epsilon,\\
&=\epsilon, 
\end{align*} 
as required.
\end{proof}
This definition of strong security for an AMD code is stronger than the corresponding notions of security against substitution against an authentication code: in (\ref{eq:subsuc}) the adversary's success probability is defined with respect to a specific source distribution, whereas for a strong AMD code we require success probability at most $\epsilon$ when attacking any possible source distribution.  However, we note that this stronger notion of security has also been considered in the context of authentication codes  \cite{Stinson88,Stinson1990}.

\section{A design-theoretic perspective on splitting authentication codes}
\label{sec:splittingsetsystem}
The notion of a {\em splitting BIBD} was introduced in \cite{Ogata2004} for the purpose of classifying splitting authentication codes that were optimal with respect to certain bounds on their parameters.  In this section we introduce the related but weaker notion of a {\em splitting set system}.  This is essentially a way of describing a splitting authentication code using design-theoretic notation that will allow us to illuminate the fundamental connection between splitting authentication codes with perfect secrecy and AMD codes, as well as describe a wider class of authentication codes with useful properties, including perfect secrecy.
\begin{definition}
A {\em $(v,b,m)$-splitting set system}  consists of a finite set $\mathcal{ V}$ of {\em points} with $|\mathcal{ V}|=v$, together with a family $\mathcal{ B}$ of {\em blocks} where $|\mathcal{ B}|=b$ and each block $B\in \mathcal{ B}$ consists of a list of $m$ pairwise disjoint subsets $(B_1,B_2,\dotsc,B_m)$, with $B_j\subset \mathcal{ V}$ for $j=1,2,\dotsc,m$.  If $v=b$ then the splitting set system is said to be {\em symmetric}.
\end{definition}
\begin{example}\label{ex:sss}\cite{Huber12}
Let $\mathcal{V}=\{0,1,2,3,4,5,6,7,8\}$.  The rows of the following array give the blocks of $(9,9,2)$-splitting set system.  In the row corresponding to a block $B\in\mathcal{B}$  the vertical line separates the points of $B_1$ from the points of $B_2:$
\begin{equation*}
\begin{array}{c|c}
0,1&3,5\\
1,2&4,6\\
2,3&5,7\\
3,4&6,8\\
4,5&7,0\\
5,6&8,1\\
6,7&0,2\\
7,8&1,3\\
8,0&2,4
\end{array}
\end{equation*}
\end{example}

Consider the special case where each of the subsets $B_j$ has size $c$.   In this setting an $(v,b,m)$-splitting set system is known as a $(v,m\times c,\lambda)$-{\em splitting balanced incomplete block design}  (splitting BIBD) if it satisfies the following condition:
\begin{itemize}
\item for every pair $P,Q\in \mathcal{ V}$ with $P\neq Q$ there are precisely $\lambda$ blocks $B$ with $P\in B_j$ and $Q\in B_{j^\prime}$  for some $j,j^\prime$ with $j\neq j^\prime$.  
\end{itemize}
We note that the splitting set system of Example~\ref{ex:sss} is a $(9,2\times 2,1)$-splitting BIBD.

Splitting BIBDs were introduced in \cite{Ogata2004}, where they were shown to be equivalent to certain optimal splitting authentication codes.  More generally, every splitting authentication code with equiprobable encoding gives rise to a splitting set system (and {\em vice versa}) by making the following identifications: 

\begin{itemize}
\item the set of points $\mathcal{ V}$ is simply the set of messages $\mathcal{ T}$ of the authentication code;
\item for each key $k \in \mathcal{ K}$ we obtain a block $B$ by letting $B_j=e_{k}(s_j)$ for $j=1,2,\dotsc m$.
\end{itemize}

Note that it may be the case that two different keys give rise to the same encodings of the sources.  In this case the splitting set system would have repeated blocks, which nonetheless correspond to distinct keys.  In what follows, however, we restrict our attention to splitting set systems without repeated blocks.

This equivalent description gives a useful language for illustrating how the combinatorial properties of an authentication code with equiprobable encoding determine its security properties.

\begin{example}
For an authentication code with equiprobable sources and equiprobable encoding we can reformulate the expression for the success probability $\epsilon_\sigma$ of an adversary's strategy $\sigma$ given in (\ref{eq:subsuc}) in terms of the language of splitting set systems.  For $B\in \mathcal{ B}$, the set $X_{B,s_j}^\sigma$ is the set of points $P\in B_j$ for which $\sigma(P)\in B_{j^\prime}$ for some $j^\prime \neq j$, and $\epsilon_\sigma$ becomes
\begin{align}
\epsilon_\sigma &=\sum_{b\in\mathcal{B}}\sum_{j=1}^m \frac{|X_{B,s_j}^\sigma|}{bm|B_{j}|}.\label{eq:subprob}
\end{align}
In the $c$-regular case, (\ref{eq:subprob}) becomes
\begin{align*}
\epsilon_\sigma &=\frac{1}{bmc}\sum_{B\in\mathcal{B}}\sum_{j=1}^m |X_{B,s_j}^\sigma|.
\end{align*} 
Now suppose our splitting set system is a $(v,m\times c,\lambda)$-splitting BIBD.  For any $P\in \mathcal{ V}$, if $\sigma(P)=P$ then $P\notin X^\sigma_{B,s_j}$ for any $i,j$.  However, for each of the $v$ points $P\in \mathcal{ V}$, if $\sigma(P)\neq P$, then there are $\lambda$ blocks  $B$ with $P\in B_{j}$ and $\sigma(P)\in B_{j^\prime}$ for some $j,j^\prime$ with $j^\prime\neq j$.  Hence
\begin{align*}
\sum_{B\in\mathcal{B}}\sum_{j=1}^m |X_{B,s_j}^\sigma|\leq \lambda v,
\intertext{ with equality occurring for any strategy $\sigma$ with $\sigma(P)\neq P$ for any $P\in \mathcal{ V}$.  Hence for any such $\sigma$ we have}
\epsilon_\sigma=\frac{\lambda v}{bmc}.
\end{align*}
In \cite{Ogata2004} it was shown that, for a $(v,m\times c,\lambda)$-splitting BIBD, we have \begin{align*}\lambda=\frac{bmc(mc-c)}{v(v-1)}.\end{align*}  Thus we can express $\epsilon_\sigma$ as $(mc-c)/(v-1)$.  We observe that this is precisely the expression given by (\ref{eq:optimalsub}) for this authentication code, and hence we see it has optimal substitution probability.
\end{example}

\begin{example}
We now determine the impersonation probability.  The expression (\ref{eq:imp1}) can be interpreted as saying an adversary who attempts impersonation by sending point $P$ succeeds with probability equal to the number of blocks that contain $P$, divided by the total number of blocks.  For a $(v,m\times c,\lambda)$-splitting BIBD, each point is contained in $\lambda (v-1)/((m-1)c)$ points \cite{Ogata2004}, and so this probability becomes
\begin{align*}
\frac{\lambda (v-1)}{(m-1)cb}.
\end{align*}
Again expressing $\lambda$ in terms of the other parameters we can rearrange this expression to determine that the impersonation probability is $mc/v$.

On the other hand, the expression for the optimal impersonation probability given in (\ref{eq:imp2}) is equivalent to the size of the smallest block divided by the number of points.  For a $(v,m\times c,\lambda)$-splitting BIBD this is simply $mc/v$, hence we see that a splitting BIBD gives rise to an authentication code with optimum impersonation probability.  (This result was proved as part of Theorem~5.5 of \cite{Ogata2004} in the case where $\lambda=1$.)
\end{example}

\subsection{Automorphism groups of splitting set systems}
We have seen that the additional structure of a splitting BIBD makes it easier to analyse the properties of the corresponding authentication codes, and to derive further results in those codes having some desirable properties. In a similar vein, we now turn our attention to the question of how the presence of certain symmetries can impact the properties of authentication codes.

\begin{definition}
An {\em automorphism} of a splitting set system is a bijection $\theta\colon\mathcal{ V}\rightarrow \mathcal{ V}$ that preserves incidence in the sense that if $B=(B_{1},B_{2},\dotsc,B_{m})\in\mathcal{ B}$ then $B^\theta=(B_{1}^\theta,B_{2}^\theta,\dotsc,B_{m}^\theta)\in\mathcal{ B}$, where $B_{j}^\theta=\{P^\theta: P\in B_{j}\}$ for $j=1,2,\dotsc,m$.
\end{definition}

\begin{definition}
We say that a splitting set system $(\mathcal{ V},\mathcal{ B})$ is {\em group generated} if there is an abelian subgroup $\mathcal{ G}$ of its automorphism group that acts regularly on $\mathcal{ V}$.  That is, $(\mathcal{ V},\mathcal{ B})$ is group generated if and only if there is an abelian subgroup $\mathcal{ G}$ of its automorphism group with the property that for every pair of points $P,Q\in \mathcal{V}$ there is precisely one element $g\in \mathcal{ G}$ such that $P^g=Q$.
\end{definition}
Note that this definition extends readily to the case where $\mathcal{ G}$ is nonabelian, but for the purposes of this paper, we restrict our attention to abelian groups.

\begin{example}
The {\em cyclic splitting designs} defined by Huber in \cite{Huber12} are examples of group generated splitting set systems.  Specifically, they can be viewed as group generated splitting set systems in which the set system is a splitting BIBD, and $\mathcal{G}$ is a cyclic group.  Huber gives a construction of a group generated splitting set system with $m=2$, $v=2c^2n+1$ and $b=(2c^2n+1)n$ for any $c\geq 1$, $n\geq 1$ with $\mathcal{G}=\mathbb{Z}_v$.  The splitting set system in Example~\ref{ex:sss} arises from this construction with $c=2$, $n=1$ and $v=9$.
\end{example}

Consider the action of $\mathcal{ G}$ on $\mathcal{ B}$ that is induced by the action of $\mathcal{ G}$ on $\mathcal{ V}$.  We refer to the orbits of blocks under this action as {\em block orbits} of the splitting set system.  To simplify the presentation and analysis, in this paper we will assume that all block orbits have size $|\mathcal{ G}|$, i.e. that $\mathcal{ G}$ acts semiregularly on $\mathcal{ B}$.\footnote{If the orbit sizes are not uniform then we cannot guarantee perfect secrecy by taking a uniform distribution on the blocks.  A closer attention to the probabilities and a careful application of the Orbit-Stabliser Theorem is required to analyse this case.}

The following lemma sets out some useful combinatorial properties of the block orbits:

\begin{lemma}\label{lem:algproperties}

Let $\Omega$ be a block orbit of a group generated $(v,b,m)$-splitting set system with $|\Omega|=v$.  Then the blocks in $\Omega$ satisfies the following properties:
\begin{enumerate}
\item For any $j=1,2,\dotsc,m$ the set $B_{j}$ has the same size for all $B\in \Omega$.  Denote this size by $|B_{j}|=c^\Omega_j$.  
\item Every block $B\in \Omega$ contains the same number of points.  Denote this number by $\ell^\Omega=\sum_{j=1}^{m} c^\Omega_j$.
\item Every point $P\in \mathcal{ V}$ occurs in $c_j^\Omega$ of the sets $B_{j}$ with $B\in \Omega$.
\item Every point $P\in \mathcal{ V}$ occurs in $\ell^\Omega$ of the blocks in $\Omega$.

\end{enumerate}
\end{lemma}

\begin{proof}
Let $\mathcal{ G}$ be the abelian subgroup of the automorphism group that acts regularly on $\mathcal{ V}$.  By definition, $\mathcal{ G}$ acts regularly on $\Omega$.  
\begin{enumerate}
\item \label{asdf} Let $B\in\Omega$, and set $c^\Omega_j=|B_{j}|$ for $j=1,2\dotsc,m$.  Since $\mathcal{ G}$ acts regularly on $\Omega$, it follows that for any $B^\prime\in \Omega$ we have $B^\prime=B^g$ for some $g\in\mathcal{ G}$ and hence $B^\prime_j=B_{j}^g$.  This implies that $|B^\prime_j|=c^\Omega_j$ for all $B^\prime\in \Omega$.
\item This follows immediately from \ref{asdf}.
\item \label{asdfasdf}The number of pairs $(P,B)$ where $B \in \Omega$ and $P\in B_{j}$ is $|\Omega|c_j^\Omega=vc_j^\Omega$.  The collection of sets $B_{j}$ with  $B \in \Omega$ is a union of orbits under the action of $\mathcal{ G}$, and hence is fixed when acted on by any element of $\mathcal{ G}$.  As $\mathcal{ G}$ acts regularly on $\mathcal{ V}$, it follows that the multiset of points contained in the multiset union $\bigcup_{B\in \Omega}B_{j}$ contains each element of $\mathcal{ V}$ an equal number of times.  For, if some element $P_1$ occurred more times than the element $P_2$, then acting on $\{B_{j} : B \in \Omega\}$ with the unique element $g\in \mathcal{ G}$ for which $P_1^g=P_2$ would not fix $\{B_{j} : B \in \Omega\}$.  Thus we conclude that each point occurs $vc_j^\Omega/v=c_j^\Omega$ times in this union, and furthermore these occurrences are all in distinct sets $B_{j}$. (By construction, no set $B_{j}$ contains repeated elements.)
\item This follows immediately from \ref{asdfasdf}.
\end{enumerate}
\end{proof}
Property \ref{asdfasdf} of Lemma~\ref{lem:algproperties}, considered together with Theorem~2.3 of \cite{stinpat19} (which applies only to $c$-splitting authentication codes) implies that a $c$-splitting authentication code arising from a group generated splitting set system has both perfect secrecy  and optimal impersonation probability.  By restricting our attention to group-generated splitting set systems we can prove results analogous to Lemma~2.2 and Theorem~2.3 of \cite{stinpat19} for authentication codes that are not necessarily $c$-splitting.
\begin{theorem}\label{thm:uniform}
The messages of an authentication code corresponding to a group generated splitting set system  are distributed uniformly, and this distribution is independent of the distribution of the sources.  
\end{theorem}
\begin{proof}
Suppose $(\mathcal{ V},\mathcal{ B})$ is a group generated $(v,b,m)$-splitting set system for which $\mathcal{ G}$ is an abelian subgroup of the automorphism group that acts regularly on $\mathcal{ V}$.  Suppose that there are $h$ block orbits, so we have $b=hv$. Fix a source $s_j$.  A point $P\in \mathcal{ V}$ occurs $c_j^\Omega$ times in sets $B_{j}$ with $B\in\Omega$.  Each of these instances arises with probability $(b{c_j^\Omega})^{-1}$, hence the total probability of obtaining the message $P$ when the source is $s_j$ and the key corresponds to a block in the orbit $\Omega$ is $b^{-1}$.  Summing over all $h$ orbits, we see that the total probability of obtaining the message $P$ when the source is $s_j$ is $hb^{-1}=v^{-1}$.   As $\Pr(P\mid s_j)=v^{-1}$ for all $P\in \mathcal{ V}$ and all $s_j\in \mathcal{ S}$ we conclude that the messages are uniformly distributed, independently of the source, as required.
\end{proof}
This result leads directly to the following corollary.
\begin{corollary}\label{cor:perfsec}
An authentication code corresponding to a group generated splitting set system has perfect secrecy. 
\end{corollary}
\begin{corollary}\label{cor:optimp}
An authentication code corresponding to a group generated $(v,b,m)$-splitting set system has optimal impersonation probability if and only if each block has the same number of points. \label{cor:optimp}
\end{corollary}
\begin{proof}
By Lemma~\ref{lem:algproperties}, the blocks occurring in some block orbit $\Omega_i$ have size $\ell^{\Omega_i}$.  Hence the impersonation probability of the authentication code is greater than or equal to $\min_{i\in \{1,2,\dotsc,h\}}{\ell^{\Omega_i}}/{v}$.

Let $P\in \mathcal{ V}$.  By Lemma~\ref{lem:algproperties}, the number of keys that can give rise to $P$ as an encoding of source $j$ is given by
\begin{align*}
\sum_{i=1}^h c_j^{\Omega_i},
\end{align*}
and the total number of keys that can give rise to $P$ as an encoding of some source is thus
\begin{align*}
\sum_{j=1}^m \sum_{i=1}^h c_j^{\Omega_i}&=\sum_{i=1}^h\ell^{\Omega_i}.
\end{align*}
 We observe that, if $\ell^{\Omega_i}$ is some constant $\ell$ for each $i$, then this expression is simply $h\ell$, and so the impersonation probability is $h\ell/(hv)=\ell/v$, which is optimal.  However, if $\ell^{\Omega_i}$ varies with $i$, then 
\begin{align*}
\sum_{i=1}^h\ell^{\Omega_i}&> \sum_{i=1}^h\min_{i\in \{1,2,\dotsc,h\}}\ell^{\Omega_i},\\
&=h\min_{i\in \{1,2,\dotsc,h\}}\ell^{\Omega_i},
\end{align*}
and so the impersonation probability is strictly greater than $\min_{i\in \{1,2,\dotsc,h\}}\ell^{\Omega_i}/v$ and hence it is not optimal.
\end{proof}
Lemma~2.2 of \cite{stinpat19} showed that a $c$-splitting authentication code for $m$ sources, $v$ messages and $b$ keys has optimal impersonation probability if and only if each message $P$ is contained in $bcm/v$ blocks $B$, which can be seen as a result that is in some sense dual to Corollary~\ref{cor:optimp} in the setting of group-generated $c$-splitting authentication codes.  We note that in the case of a $c$-splitting group-generated splitting set system, each point is contained in $cm$ of the blocks in each orbit  by Lemma~\ref{lem:algproperties}.  Summing over all $h$ orbits, this implies that each point is contained in a total of $hcm=bcm/v$ blocks, and hence the fact that it has optimal impersonation probability follows from Lemma~2.2 of \cite{stinpat19}.

These results show that group-generated splitting set systems give rise to a class of authentication codes with interesting and useful properties.  They are also  a natural class to consider from a point of view of seeking good constructions of splitting authentication codes: the literature contains examples of group-generated splitting BIBDs such as those arising from external difference families \cite{Ogata2004}, and those in \cite{Wang,WS, Huber12}.  (We observe that as they are group-generated, the splitting authentication codes constructed in \cite{Wang} provide perfect secrecy even though this property is not considered in that paper.)

\subsection{AMD codes and group generated splitting set systems}
In this section we explore a close connection between group-generated splitting set systems and AMD codes that allows us to view an AMD code as a special case of an authentication code with perfect secrecy.  We consider the cases of weak and strong AMD codes separately; we will see that these correspond respectively to authentication codes that require a uniform source distribution, or to those permit any distribution on the sources.

Starting with a weak AMD code, we can obtain a splitting set system by constructing its {\em development}.  (This can be seen as a generalisation of Theorem~3.4 of \cite{Ogata2004}.)    
\begin{definition} Let $A(s_1),A(s_2),\dotsc,A(s_m)\subset \mathcal{ G}$ be the sets of valid encodings of the sources of a weak AMD code.  The {\em development} of this AMD code is the splitting set system obtained by setting $\mathcal{ V}=\mathcal{ G}$, and letting $\mathcal{ B}$ be the set of all blocks of the form $(g+A(s_1),g+A(s_2),\dotsc,g+A(s_m))$  for some $g\in \mathcal{ G}$.
\end{definition}
(Note that to simplify the presentation and analysis, we restrict our attention to the case where the development contains $|\mathcal{ G}|$ distinct blocks.)  This construction allows us to interpret a weak AMD code as a traditional authentication code:  

\begin{theorem}\label{thm:devisacode}
The development of a weak $(m,n,\epsilon)$-AMD code $(\mathcal{ S},\mathcal{ G},\mathcal{ A},\mathcal{ E})$ with equiprobable encoding is a group generated $(n,b,m)$ splitting set system $(\mathcal{ V},\mathcal{ B})$.  In the case where it has $n$ distinct blocks, it has $\mathcal{ G}$ as a subgroup of its automorphism group that acts regularly on $\mathcal{ V}$ and on $\mathcal{ B}$. The corresponding splitting authentication code has perfect secrecy and optimal impersonation probability and its substitution probability is at most $\epsilon$ when the sources are chosen uniformly.
\end{theorem}
\begin{proof}
By construction, the development of the $(m,n,\epsilon)$-AMD code $(\mathcal{ S},\mathcal{ G},\mathcal{ A},\mathcal{ E})$ is a splitting set system $(\mathcal{ V},\mathcal{ B})$ whose points are the elements of $\mathcal{ G}$; hence, the number of points is $n$.  Again, by construction we see that addition by an element of $\mathcal{ G}$ gives an automorphism of $(\mathcal{ V},\mathcal{ B})$.  Since $\mathcal{ G}$ acts regularly on itself by addition it follows that $\mathcal{ G}$ is a subgroup of the automorphism group of $(\mathcal{ V},\mathcal{ B})$ that acts regularly on $\mathcal{ V}$, hence the splitting set system is group generated.  The blocks of $\mathcal{ B}$ lie in a single orbit, so if $|\mathcal{ B}|=n$ then the action of $\mathcal{ G}$ on $\mathcal{ B}$ is regular.  Corollary~\ref{cor:perfsec} shows that the corresponding splitting authentication code has perfect secrecy, and Corollary~\ref{cor:optimp} shows that it has optimal impersonation probability.  

We now determine the substitution probability of the authentication code corresponding to  $(\mathcal{ V},\mathcal{ B})$.  Let $B^0$ denote the block $(A(s_1),A(s_2),\dotsc,A(s_m))$.  Consider a substitution strategy $\sigma$. We have
\begin{align*}
\epsilon_\sigma&=\sum_{j=1}^m\frac{1}{mb|A(s_j)|}\sum_{B\in\mathcal{B}}|X^\sigma_{B,s_j}|.
\end{align*}

The sum $\sum_{B\in\mathcal{B}}|X^\sigma_{B,s_j}|$ counts all pairs $(P,B)\in\mathcal{ V}\times\mathcal{ B}$ with $P\in B_{j}$ and $\sigma(P)\in B_{j^\prime}$ for some $j^\prime \neq j$. 
We compute this expression in a different way.  Let $P\in \mathcal{ V}$, and set $\Delta_P=\sigma(P)-P$.  For each point $Q\in B^0_{j}$ there is a unique element $g_Q\in \mathcal{ G}$ with $Q+g_Q=P$.  Let $B^Q$ denote the block $g_Q+B^0$.  Then $P\in B^Q_{j}$.  We claim that $P\in X_{B^Q,s_j}^\sigma$ if and only if $Q\in X_{s_j}^{\Delta_P}$: when $Q\in X_{s_j}^{\Delta_P}$ we have $Q+\Delta_P\in B^0_{j^\prime}$ for some $j^\prime \neq j$, in which case the element 
\begin{align*}
g_Q+(Q+\Delta_P)&=(g_Q+Q)+\sigma(P)-P,\\
&=P+\sigma(P)-P,\\
&=\sigma(P)
\end{align*} 
lies in $B^Q_{j^\prime}$ (Figure~\ref{fig:onlyfig}).  

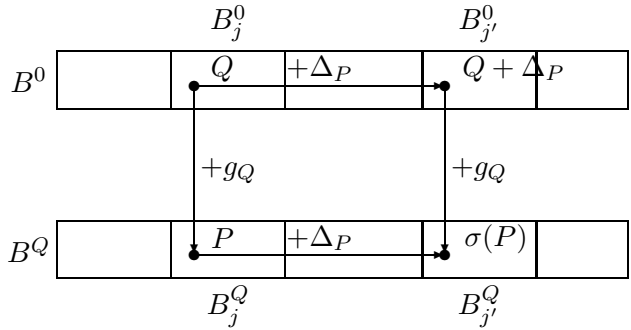
\begin{figure}[h!]
\begin{center}
\setlength{\unitlength}{1.5mm}
\begin{picture}(70,35)
\put(5,5){\framebox(50,5){}}
\put(5,20){\framebox(50,5){}}
\put(0,5){\makebox(5,5){$B^Q$}}
\put(0,20){\makebox(5,5){$B^0$}}
\put(15,5){\line(0,1){5}}
\put(25,5){\line(0,1){5}}
\put(37,5){\line(0,1){5}}
\put(47,5){\line(0,1){5}}
\put(15,20){\line(0,1){5}}
\put(25,20){\line(0,1){5}}
\put(37,20){\line(0,1){5}}
\put(47,20){\line(0,1){5}}
\put(17,7){\circle*{1}}
\put(39,7){\circle*{1}}
\put(39,22){\vector(0,-1){15}}
\put(17,22){\vector(0,-1){15}}
\put(17,7){\vector(1,0){22}}
\put(17,22){\vector(1,0){22}}
\put(17,22){\circle*{1}}
\put(39,22){\circle*{1}}
\put(17,6){\makebox(5,5){$P$}}
\put(17,21){\makebox(5,5){$Q$}}
\put(39,21){\makebox(12,5){$Q+\Delta_P$}}
\put(39,6){\makebox(9,5){$\sigma(P)$}}
\put(17,12){\makebox(6,5){$+g_Q$}}
\put(39,12){\makebox(6,5){$+g_Q$}}
\put(15,25){\makebox(10,5){$B^0_{j}$}}
\put(15,0){\makebox(10,5){$B^Q_{j}$}}
\put(37,25){\makebox(10,5){$B^0_{j^\prime}$}}
\put(37,0){\makebox(10,5){$B^Q_{j^\prime}$}}
\put(25,21){\makebox(6,5){$+\Delta_P$}}
\put(25,6){\makebox(6,5){$+\Delta_P$}}
\end{picture}
\end{center}
\caption{$P\in X_{B^Q,s_j}^\sigma$ when $Q\in X_{s_j}^{\Delta_P}$}
\label{fig:onlyfig}
\end{figure}

Thus the point $P$ lies in precisely $|X_{s_j}^{\Delta_P}|$ of the sets $X_{B,s_j}^\sigma$ and we conclude that 
\begin{align*}
\sum_{B\in\mathcal{B}}|X^\sigma_{B,s_j}|&=\sum_{P\in G}|X^{\Delta_P}_{s_j}|,
\intertext{which implies that the success probability of $\sigma$ is}
\sum_{j=1}^m\frac{1}{mn|A(s_j)|}\sum_{P\in \mathcal{ G}}|X^{\Delta_P}_{s_j}|&=\frac{1}{n}\sum_{P\in \mathcal{ G}}\left(\sum_{j=1}^m \frac{|X^{\Delta_P}_{s_j}| }{m|A(s_j)|}\right),\\
&\leq \frac{1}{n}\sum_{P\in \mathcal{ G}}\epsilon,\\
&=\epsilon,
\end{align*}
since an adversary who chooses $\Delta_P$ in attacking the AMD code has success probability at most $\epsilon$.

\end{proof}
The results of Theorem~\ref{thm:devisacode} show that this interpretation is both useful and natural, as they demonstrate that the greatest success probability of any adversary in attacking the AMD code corresponds directly to the greatest success probability of any substitution strategy for attacking the authentication code.  Immediate consequences of this interpretation include the fact that known bounds on the parameters of authentication codes apply also to the parameters of AMD codes.  Indeed bounds from the literature on the parameters of AMD codes can be seen to be special cases of existing bounds for authentication codes.  Further consequences of this connection will be discussed later in this section.

The following result can be seen as a weak converse of Theorem~\ref{thm:devisacode}:
\begin{theorem}\label{thm:acodetoamd} A group generated $(v,v,m)$ splitting set system $(\mathcal{ V},\mathcal{ B})$ with a single block orbit of size $v$, for which the corresponding authentication code has substitution probability at most $\epsilon$ when the sources are uniformly distributed, gives rise to a weak $(m,v,\epsilon)$-AMD code.
\end{theorem}
\begin{proof}
Let $\mathcal{ G}$ be the subgroup of the automorphism group of $(\mathcal{ V},\mathcal{ B})$ that acts regularly on $\mathcal{ V}$.    Fix a point $P$.  For every point $P^\prime \in \mathcal{ V}$, there is a unique element $g$  of $\mathcal{ G}$ for which $P^g=P^\prime$; we can thus identify these points with these group elements.  Pick a block $B^0\in \mathcal{ B}$; then the sets $B^0_{j}\subset\mathcal{ G}$, for $j=1,2,\dotsc, m$, are pairwise disjoint and can be regarded as the sets $A(s_1),A(s_2),\dotsc, A(s_m)$ of an AMD code $\mathcal{ A}$.  As there is a single block orbit of size $v$, $\mathcal{ G}$ acts regularly on $\mathcal{ B}$ and so we observe that $(\mathcal{ V},\mathcal{ B})$ is in fact the development of $\mathcal{ A}$.  Consider a substitution strategy $\sigma_\Delta$ defined by setting $\sigma(P)=P+\Delta$ for each $P\in \mathcal{ V}$.  We observe that for any $i\in 1,2,\dotsc, v$ we have $|X^{\sigma_\Delta}_{B,s_j}|=|X_{s_j}^\Delta|$ by construction, since addition of group elements preserves differences.  Thus we have
\begin{align*}
\epsilon_{\sigma_\Delta}&=\sum_{j=1}^m\frac{1}{mvc_j}\sum_{B\in\mathcal{B}}|X^{\sigma_\Delta}_{B,s_j}|,\\
&=\sum_{j=1}^m\frac{1}{mv|A(s_j)|}\sum_{B\in\mathcal{B}}|X^\Delta_{s_j}|,\\
&=\frac{1}{v}\sum_{B\in\mathcal{B}} \epsilon_\Delta,\\
&=\epsilon_\Delta.
\end{align*} 
Hence we see that $\mathcal{ A}$ is a weak $(m,v,\epsilon)$-AMD code, as $\epsilon_{\sigma_\Delta}\leq \epsilon$ for all $\Delta\in \mathcal{ G}$.
\end{proof}
This correspondence is not specific to the case of weak AMD codes: if we replace the weak AMD code by a strong AMD code, we obtain an authentication code that works for any source distribution.
\begin{theorem}
The development of a strong $(m,n,\epsilon)$-AMD code $(\mathcal{ S},\mathcal{ G},\mathcal{ A},\mathcal{ E})$ with equiprobable encoding is a group generated $(n,b,m)$ splitting set system $(\mathcal{ V},\mathcal{ B})$.  In the case where it has $n$ distinct blocks, it has $\mathcal{ G}$ as a subgroup of its automorphism group that acts regularly on $\mathcal{ V}$ and on $\mathcal{ B}$.  The corresponding splitting authentication code has perfect secrecy and it has substitution probability at most $\epsilon$ for any source distribution.
\end{theorem}
\begin{proof}
The proof follows that of Theorem~\ref{thm:devisacode} exactly, except for the determination of the substitution probability.  In this case we have
\begin{align*}
\epsilon_\sigma&=\sum_{B\in\mathcal{B}}\sum_{j=1}^m\frac{\Pr(s_j)|X_{B,s_j}^\sigma|}{n|A(s_j)|} .
\intertext{As before, we have $\sum_{B\in\mathcal{B}}|X_{B,s_j}^\sigma|=\sum_{P\in \mathcal{ G}}|X_{s_j}^{\Delta_P}|,$ and so}
\epsilon_\sigma&=\sum_{j=1}^m\frac{\Pr(s_j)}{n|A(s_j)|}  \sum_{P\in \mathcal{ G}}|X_{s_j}^{\Delta_P}|,\\
&=\sum_{j=1}^m\frac{\Pr(s_j)}{n}\sum_{P\in \mathcal{ G}}\epsilon_{s_j,\Delta_P},\\
&\leq\sum_{j=1}^m\frac{\Pr(s_j)}{n}\sum_{P\in \mathcal{ G}}\epsilon,\\
&=\epsilon.
\end{align*}
\end{proof}
\begin{theorem}
A group generated $(v,v,m)$ splitting set system $(\mathcal{ V},\mathcal{ B})$ with a single block orbit of size $v$, for which the corresponding authentication code has substitution probability at most $\epsilon$ for any source distribution, gives rise to a strong $(m,v,\epsilon)$-AMD code.
\end{theorem}
\begin{proof}
Consider the source distribution in which source $s_j$ is chosen with probability 1, and define the substitution strategy $\sigma_\Delta$ as in the proof of Theorem~\ref{thm:acodetoamd}.  In this case we have
\begin{align*}
\epsilon_{\sigma_\Delta}&=\frac{1}{vc_j}\sum_{B\in\mathcal{B}}|X^{\sigma_{\Delta}}_{B,s_j}|,\\
&=\frac{1}{v|A_(s_j)|}\sum_{B\in\mathcal{B}}|X^{\Delta}_{s_j}|,\\
&=\epsilon_{s_j,\Delta}.
\end{align*}
Thus, for any choice of source $s_j$, we have $\epsilon_{s_j,\Delta}=\epsilon_{\sigma_\Delta}\leq \epsilon$ and so the AMD code is a strong $(m,v,\epsilon)$-AMD code, as required.
\end{proof}
In \cite{oldamd}, the notion of an {\em R-optimal} weak (resp.\ strong) AMD code was defined.  These are  weak (resp.\ strong) AMD codes for which the success probability of the worst-case adversarial choice of $\Delta$ is equal to that of the average-case choice.  A $c$-regular weak or strong $(m,n,\epsilon)$-AMD code is R-optimal if $\epsilon=c(m-1)/(n-1)$ (see \cite{oldamd}).  
\begin{corollary}
The development of a $c$-regular R-optimal weak $(m,n,\epsilon)$-AMD code has optimal substitution probability when the sources are uniformly distributed.  The development of a $c$-regular R-optimal strong $(m,n,\epsilon)$-AMD code has optimal substitution probability for any source distribution.
\end{corollary}
\begin{proof}
Let $\mathcal{ A}$ be an R-optimal weak $(m,n,\epsilon)$-AMD code.  By Theorem~\ref{thm:devisacode} we know that its development is an authentication code with substitution probability $c(m-1)/(n-1)$ when the sources are chosen uniformly. For this authentication code, the expression in (\ref{eq:optimalsub}) is also equal to $c(m-1)/(n-1)$, and hence we conclude that this substitution probability is optimal.
\end{proof}
In fact the bound in (\ref{eq:optimalsub}) is not tight in general for authentication codes that are not $c$-splitting, as shown by the following example:
\begin{example}
Consider the weak $(4,10,1/2)$-AMD code $\mathcal{ A}$ in $\mathbb{Z}_{10}$ that is defined by the sets $A(1)=\{0\},A(2)=\{5\},A(3)=\{1,9\},A(4)=\{2,3\}$.  This was shown in \cite{oldamd} to be R-optimal.  However, we note that the expression of (\ref{eq:optimalsub}) for the corresponding authentication code is 
\begin{align*}
\frac{6-2}{10-1}=\frac{4}{9}<\frac{1}{2}.
\end{align*}
\end{example}

Inspired by the R-bound \cite{oldamd} for AMD codes, we now establish a bound on the substitution probability for splitting authentication codes that coincides with (\ref{eq:optimalsub}) in the $c$-splitting case, but which is tighter for authentication codes that are not $c$-splitting.  Note that, although we make use of a cyclic group in the proof, we are not assuming that the authentication code is group generated, since the group elements are, in general, not automorphisms of the authentication code.

\begin{theorem}\label{thm:newbound}
Let $(\mathcal{ V},\mathcal{ B})$ be a $(v,b,m)$-splitting set system arising from an authentication code with substitution probability $\epsilon$. Then
\begin{align*}
\epsilon\geq \sum_{B\in\mathcal{B}} \frac{1}{b}\left(\frac{|B|-\sum_{j=1}^m\Pr(s_j)|B_{j}|}{v-1}\right).
\end{align*}
\end{theorem}
\begin{proof}
We prove this result by showing that there exists a substitution strategy whose success probability is at least this value.

We identify the points of $\mathcal{ V}$ with the elements of $\mathbb{Z}_{v}=\{0,1,2\dotsc,v-1\}$.  For $r\in\{1,2,\dotsc,v-1\}$ we define a substitution strategy $\sigma_r$ by setting $\sigma_r(P)=P+r \pmod{v}$ for all $P\in \mathcal{ V}$.  We now compute the mean $\overline{\epsilon_{\sigma_r}}=\sum_{r=1}^{v-1}\frac{1}{v-1} \epsilon_{\sigma_r}$ as follows:
\begin{align*}
\overline{\epsilon_{\sigma_r}}&=\sum_{r=1}^{v-1}\frac{1}{v-1} \epsilon_{\sigma_r},\\
&=\frac{1}{v-1}\sum_{r=1}^{v-1}\sum_{B\in\mathcal{B}}\sum_{j=1}^m\frac{|X^{\sigma_r}_{B,s_j}|\Pr(s_j)}{b|B_{j}|},\\
&=\frac{1}{(v-1)b}\sum_{j=1}^m\Pr(s_j) \sum_{b\in\mathcal{B}}\frac{1}{|B_{j}|} \sum_{r=1}^{v-1}|X^{\sigma_r}_{B,s_j}|.\\
\intertext{Consider the set $X^{\sigma_r}_{B,s_j}$.  There are $|B_{j}|$ elements in $B_{j}$ and $|B|-|B_{j}|$ elements in $\bigcup_{j^\prime \neq j}
B_{j^\prime}$.  For each pair of elements $P\in B_{j}$ and $Q\in B_{j^\prime}$ with $j^\prime \neq j$ there is a unique value of $r$ in $\{1,2,\dotsc,r-1\}$ with $\sigma_{r}(P)=Q$.  In this case we thus have $P\in X^{\sigma_r}_{B,s_j}$.  Hence we see that $\sum_{r=1}^{v-1}|X^{\sigma_r}_{B,s_j}|$ is equal to the number of such pairs, which is $|B_{j}|(|B|-|B_{j}|)$. Thus we have}    
\overline{\epsilon_{\sigma_r}}&=\frac{1}{(v-1)b}\sum_{j=1}^m\Pr(s_j) \sum_{b\in\mathcal{B}}\frac{1}{|B_{j}|} |B_{j}|(|B|-|B_{j}|),\\
&=\frac{1}{(v-1)b}\sum_{j=1}^m\Pr(s_j) \sum_{b\in\mathcal{B}}(|B|-|B_{j}|),\\
&=\frac{1}{(v-1)b}\sum_{b\in\mathcal{B}}\sum_{j=1}^m\Pr(s_j)(|B|-|B_{j}|),\\
&=\frac{1}{(v-1)b}\sum_{b\in\mathcal{B}} \left(|B|-\sum_{j=1}^m\Pr(s_j)|B_{j}|\right).
\end{align*}
Since this quantity is the mean of the success probabilities $\epsilon_{\sigma_r}$, we conclude that there is at least one value of $r\in\{1,2,\dotsc, v-1\}$ for which $\epsilon_{\sigma_r}$ is greater than or equal to this quantity.
\end{proof}
We note that the quantity $\sum_{j=1}^m\Pr(s_j)|B_{j}|$ is the average, over all sources $s_j$, of the size of the set $B_{j}$ of possible encodings of $s_j$ when the key is $B$.  The corresponding bound in \cite{Blundo99} has instead the maximum over all sources $s_j$ of the size of $B_{j}$.  For authentication codes that are not $c$-splitting, this new bound is thus tighter.  This new bound now corresponds directly to the R-bound for an AMD code, in both the weak and strong cases.  R-optimal AMD codes can be viewed as those where the success probability of the worst case choice of $\delta$ (i.e.\ the most successful $\delta$) is equal to that of the average case (so that in fact the success probability of each choice of $\delta$ is the same.)  A similar interpretation holds for this new bound, making it a rather natural one:

\begin{theorem}
Let $(\mathcal{ V},\mathcal{ B})$ be a $(v,b,m)$-splitting set system arising from an authentication code whose substitution probability $\epsilon$ attains the bound of Theorem~\ref{thm:newbound}.  Then any substitution strategy $\sigma$ for which $\sigma(P)\neq P$ for all $P\in \mathcal{ V}$ has $\epsilon_\sigma=\epsilon$.
\end{theorem}
\begin{proof}
The value $\overline{\epsilon_{\sigma_r}}$ was an average value taken over the $v-1$ substitution strategies of the form $\sigma_r$.  We need to show that no other substitution strategies can be more successful.  The key thing to note about the set of substitution strategies $\{\sigma_1,\sigma_2,\dotsc,\sigma_{v-1}\}$ is that for each pair $P,Q\in \mathcal{ V}$ with $P\neq Q$ there is precisely one strategy $\sigma_r$ in that set with $\sigma_r(P)=Q$.  Suppose instead that we wish to calculate the average success probability over the set $\Gamma$ of all strategies $\sigma$ for which $\sigma(P)\neq P$ for all $P\in \mathcal{ V}$.  For any pair $P,Q\in \mathcal{ V}$ with $P\neq Q$, there are $(v-1)^{v-1}$ elements $\sigma\in \Gamma$ with $\sigma(P)=Q$.  Note that $|\Gamma|=(v-1)^v$.  So when we repeat the calculation we did before, the average will turn out the same, since the sum will be $(v-1)^{v-1}$ times larger, but we are dividing by $(v-1)^v$ instead of by $(v-1)$.  The bound is only tight if the worst-case probability is equal to the average case, which implies that all the strategies in $\Gamma$ are equiprobable.
\end{proof}
Observe in addition that the bound will only be tight if $|B|-\sum_{j=1}^m\Pr(s_j)|B_{j}|$ is constant, independent of $i$.  For uniform sources this becomes $(m|B|-1)/m$, so we require the size of the blocks to be the same.

\section{Constructions for Authentication Codes with Splitting and Perfect Secrecy}
We observe that the literature contains examples of group-generated authentication codes that are not AMD-codes, for example splitting BIBDs that are the development of more than one base block \cite{Wang,WS}.  Group-generated splitting BIBDs have perfect secrecy by Corollary~\ref{cor:perfsec}.  In this section we consider a combinatorial property that allows us to determine when a splitting BIBD has perfect security.
\subsection{Equitably Ordered Splitting BIBDs}
\label{sec:equitableordering}
We recall from Section~\ref{sec:splittingsetsystem} that a \emph{$(v, m \times c,1)$-splitting BIBD} is a set system consisting of a set $\mathcal{ V}$ of  $v$ points and a set $\B$ of blocks of size $mc$, which satisfies the following properties:
\begin{enumerate}
\item each block $B$ can be partitioned into $m$ subsets of size $c$, which are denoted $B_j$, $1 \leq j \leq m$, and 
\item given any two distinct points $x$ and $y$, there is a unique block $B$ such that $x \in B_j$ and $y \in B_j^\prime$, where $j \neq j^\prime$.
\end{enumerate}

A $(v, m \times c,1)$-splitting BIBD has replication number $r$ and $b$ blocks, where
\begin{eqnarray*}
 r & = & \frac{v-1}{(m-1)c} \quad \text{and}\\
 b & = & \frac{v(v-1)}{m(m-1)c^2}.
\end{eqnarray*}
Of course $r$ and $b$ must be integers if a  $(v, m \times c,1)$-splitting BIBD exists.

Splitting BIBDs were defined in \cite{Ogata2004} as a method of constructing authentication codes with splitting.
They have been studied in a number of research papers since then.  Here, our interest is in constructing  authentication codes with splitting that also provide perfect secrecy. This can be accomplished if the splitting BIBD satisfies an additional property.

A $(v, m \times c,1)$-splitting BIBD is  \emph{equitably ordered} if, for every point $x$ and for all integers $j$ such that $1 \leq j  \leq m$, the number of blocks $B$ such that $x \in B_j$ is independent of $j$.
\begin{theorem}
 If a splitting BIBD is equitably ordered, then it yields an 
 authentication code with perfect secrecy.
\end{theorem} 
The obvious  necessary condition for a splitting BIBD to be equitably ordered is that 
\[ r \equiv 0 \bmod m.\]
Now, assuming an equitable ordering, we have $v = r (m-1)c + 1$ and $r = tm$ for some integer $t$, so
\[ v = tm (m-1)c + 1.\]
Then
\begin{eqnarray*}
 b &=& \frac{(tm (m-1)c + 1)(tm (m-1)c)}{m(m-1)c^2}\\
 &=& \frac{(tm (m-1)c + 1)t}{c}\\
 &=& t^2m(m-1) + \frac{t}{c},
 \end{eqnarray*}
 so $t =cs$, $r = csm$ and
 \[v = sm (m-1)c^2 + 1,\]
 for some integer $s$.
 That is,  a splitting BIBD can be equitably ordered only if
\begin{equation}
\label{v.eq} v \equiv 1 \bmod (m (m-1)c^2).\end{equation}

It is easy to  obtain $(v, m \times c,1)$-splitting BIBDs that can be equitably ordered if they are generated by base blocks over an abelian group. 

\begin{lemma}
\label{group.lem}
Suppose that a $(v, m \times c,1)$-splitting BIBD is generated by base blocks over an abelian group of order $v$, and suppose every orbit of blocks has size $v$. Then the splitting BIBD can be equitably ordered.
\end{lemma}

\begin{proof}
Under the stated hypotheses, the splitting BIBD is generated from \[\frac{v-1}{m(m-1)c^2}\] base blocks. Each base block gives rise to $v$ blocks in the design. We can arbitrarily order each base block. Then the development of each base block yields exactly $c$ copies of each point in each of the $m$ sets. Therefore, we get \[ \frac{v-1}{m(m-1)c} = \frac{r}{m}\] copies of each point in each of the $m$ sets.
\end{proof}

\begin{example}
\label{25,3,2.ex} A $(25, 3 \times 2,1)$-splitting BIBD is presented in \cite{GMW}. It has points in $\zed_{25}$ and it is generated from the base block 
\[ \{ \{0,1\}, \{2,4\}, \{12,20\} \}.\]
If we order the base block as \[ ( \{0,1\}, \{2,4\}, \{12,20\} )\] and maintain this ordering as the block is developed, we obtain the blocks
\[
\begin{array}{c}
( \{0,1\}, \{2,4\}, \{12,20\} ) \\
( \{1,2\}, \{3,5\}, \{13,21\} ) \\
\vdots \\
( \{24,0\}, \{1,3\}, \{11,19\} ).
\end{array}
\]
Then each point occurs twice in the union of the first sets, second sets and third sets.
\end{example}

Using the technique of Lemma \ref{group.lem}, we can construct equitably ordered $(v, 2 \times c,1)$-splitting BIBDs.

\begin{theorem}
For any $c \geq 2$, an equitably ordered $(v, 2 \times c,1)$-splitting BIBD exists if and only if 
$v \equiv 1 \bmod (2c^2)$.
\end{theorem}

\begin{proof}
Necessity follows from (\ref{v.eq}). For sufficiency, we use a construction from \cite{GMW}. There, it is shown that
a $(2c^2t+1, 2 \times c,1)$-splitting BIBD can be constructed from $t$ base blocks defined over
$\zed_{2c^2t+1}$. Applying Lemma \ref{group.lem}, we have the desired result.
\end{proof}

 We  use the following general recursive approach to construct various families of $(v, m \times c,1)$-splitting BIBDs that are equitably ordered. This construction will make use of group divisible designs.  We note that the term ``group'' here is a historical usage that does not refer to an algebraic group.  To avoid confusion, we will refer to the groups of a group divisible design as ``design groups'' to clarify that we are talking about particular sets of points in the design, rather than an algebraic group. A {\em group-divisible design} consists of a set of points $\mathcal{ V}$, a set $G$ of {\em groups} that forms a partition of $\mathcal{ V}$, and a set of blocks $\mathcal{ B}$ such that no block contains more than one point from the same design group, and every pair of points from different design groups is in a unique block.  A group divisible design is an $m$-GDD if every block has size $m$.  The {\em type} of a GDD is the multiset of its design group sizes.  The type of a GDD is usually described using an exponential notation.

 Suppose that there is an $m$-GDD on $sm(m-1)c$ points, such that
 \[ |G|  \equiv 0 \bmod (m(m-1)c) \] for every design group $G$.
 The replication number $r_x$ of any point $x \in G$ is
\begin{eqnarray*} r_x &=& \frac{sm(m-1)c - |G|}{m-1} \\
&=& \frac{sm(m-1)c - s'm(m-1)c}{m-1} \quad \text{for some intger $s'$}\\
& = &  mc(s-s')\\
& \equiv & 0\bmod m.
\end{eqnarray*}
 
 We show that an $m$-GDD satisfying the above properties can be equitably ordered, by using a technique from \cite{Stinson1990}. We first construct the point vs block bipartite incidence graph for the $m$-GDD. Each ``block'' vertex has degree $m$ and each ``point'' vertex $x$ has degree $r_x \equiv 0 \bmod m$. Split each ``point'' vertex $x$ into $r_x/m$ vertices of degree $m$. Now we have an $m$-regular bipartite graph, which therefore can be $m$-edge-coloured. Say the colours are $1,2, \dots, m$. For each block, this specifies an ordering of the points in such a way that every point occurs equally often in each position. Therefore the blocks of the GDD have been equitably ordered.
 
Next, we take $c$ copies of every point in the GDD and replace every (ordered) block by the trivial $(mc, m \times c,1)$-splitting GDD of type $c^m$. That is, each ordered block $(x_1, x_2,\dots , x_m)$ is replaced by 
\[(\{x_1\} \times \{1, \dots , c \}, \{x_2\}\times \{1, \dots , c \},\dots , \{x_m\}\times \{1, \dots , c \}).\]
This yields an $(sm(m-1)c^2, m \times c,1)$-splitting GDD that is equitably ordered.
 
Suppose further that there is a $(c|G|+1, m \times c,1)$-splitting BIBD that is equitably ordered, for every design group $G$ in the $m$-GDD. Note that \[c|G|+1 \equiv 1 \bmod (m(m-1)c^2),\] so the necessary numerical condition
(\ref{v.eq})  is satisfied. Then we obtain a $(v, m \times c,1)$-splitting BIBD  by simply taking  the blocks in the $(sm(m-1)c^2, m \times c,1)$-splitting GDD along with all the blocks in the various $(c|G|+1, m \times c,1)$-splitting BIBDs. Since each of these designs is equitably ordered, the resulting $(v, m \times c,1)$-splitting BIBD is equitably ordered.
 
 Summarizing the discussion above, we have the following.
 
\begin{theorem}
\label{general.thm}
Suppose that $v = sm (m-1)c^2 + 1$ and 
 suppose there is an $m$-GDD on $(v-1)/c$ points, 
 such that  the following conditions hold for every design group $G$:
 \begin{enumerate}
 \item $|G|  \equiv 0 \bmod (m(m-1)c)$ and
 \item there is a $(c|G|+1, m \times c,1)$-splitting BIBD that is equitably ordered.
 \end{enumerate}
 Then there is a $(v, m \times c,1)$-splitting BIBD that is equitably ordered.
\end{theorem}

We now construct several families of equitably ordered $(v, m \times c,1)$-splitting BIBDs, for fixed $m$ and $c$, using Theorem \ref{general.thm}.

\begin{theorem}
There exists a $(v, 3 \times 2,1)$-splitting BIBD that is equitably ordered if and only if
$v \equiv 1 \bmod 24$.
\end{theorem}

\begin{proof}
The necessary condition $v \equiv 1 \bmod 24$ follows from (\ref{v.eq}). We prove sufficiency using the same approach as \cite{GMW}.
Let $v = 24s+1$. For the case $s=1$, an equitably ordered $(25, 3 \times 2,1)$-splitting BIBD (from \cite{GMW}) was presented in Example \ref{25,3,2.ex}.
For $s=2$, the $(49, 3 \times 2,1)$-splitting BIBD presented in \cite{GMW} can be equitably ordered by Lemma \ref{group.lem}. For $s\geq 3$, we proceed as follows.
A $3$-GDD of type $12^s$ exists for all $s \geq 3$. As we have already mentioned, there is an equitably ordered $(25, 3 \times 2,1)$-splitting BIBD. Therefore, from Theorem \ref{general.thm}, 
we obtain a $(24s+1, 3 \times 2,1)$-splitting BIBD that is equitably ordered.
\end{proof}

\begin{theorem}
There exists a $(v, 4 \times 2,1)$-splitting BIBD that is equitably ordered if and only if
$v \equiv 1 \bmod 48$, with the possible exception of  $v = 49$.
\end{theorem}

\begin{proof}
The necessary condition $v \equiv 1 \bmod 48$ follows from (\ref{v.eq}). 
Let $v = 48s+1$. 
A $(49, 4 \times 2,1)$-splitting BIBD is not known to exist, so we cannot handle the case $s=1$. 
For $s = 2,3,4,5,6,7$ and $9$, $(48s+1, 3 \times 2,1)$-splitting BIBDs are given in 
\cite{GMW} that are generated from base blocks over groups. Therefore, using Lemma \ref{group.lem}, we have
equitably ordered splitting BIBDs for these values of $s$.

For $s=8$ and $s \geq 10$, we use $4$-GDDs on $24s$ points with design group sizes divisible by, and greater than, $24$. 
A $4$-GDD of type $48^{s/2}$ exists for all even $s \geq 8$ (see \cite{BSH}). A $4$-GDD of type $12^{(s-3)/2}18^1$ exists for all odd $s \geq 11$ (see \cite{GL}). Giving weight $4$ to every point and applying the Fundamental GDD Construction (\cite{handbookCD} Section~IV.2.1), we obtain a $4$-GDD of type $48^{(s-3)/2}72^1$  for all odd $s \geq 11$.  Hence, from Theorem \ref{general.thm}, we obtain a 
$(48s+1, 4 \times 2,1)$-splitting BIBD that is equitably ordered, for $s=8$ and for all $s \geq 10$. 
\end{proof}

\begin{theorem}
There exists a $(v, 3 \times 3,1)$-splitting BIBD that is equitably ordered if and only if
$v \equiv 1 \bmod 54$, with the possible exception of  $v = 55$.
\end{theorem}

\begin{proof} For $(v, 3 \times 3,1)$-splitting BIBDs, this result was shown by Wang \cite{Wang}. 
We use a slightly different recursive construction to construct splitting BIBDs that are equitably ordered. 
First, the necessary condition $v \equiv 1 \bmod 54$ follows from (\ref{v.eq}). 
Let $v = 54s+1$. 
A $(55, 3 \times 3,1)$-splitting BIBD is not known to exist, so we cannot handle the case $s=1$. 
For $s = 2,3,4,5$ and $7$, $(54s+1, 3 \times 2,1)$-splitting BIBDs are given in 
\cite{Wang} that are generated from base blocks over groups. Therefore, using Lemma \ref{group.lem}, we have
equitably ordered splitting BIBDs for these values of $s$.

For $s=6$ and $s \geq 8$, we use $4$-GDDs on $18s$ points with design group sizes divisible by, and greater than, $18$. 
A $3$-GDD of type $36^{s/2}$ exists for all even $s \geq 6$, and a $3$-GDD of type $36^{(s-3)/2}54^1$ exists for all odd $s \geq 9$ (see \cite{CHR}). Hence, from Theorem \ref{general.thm}, we obtain a 
$(54s+1, 3 \times 3,1)$-splitting BIBD that is equitably ordered, for $s=6$ and for all $s \geq 8$. 
\end{proof}

\begin{theorem}
There exists a $(v, 3 \times 4,1)$-splitting BIBD that is equitably ordered if and only if
$v \equiv 1 \bmod 96$.
\end{theorem}

\begin{proof}
The necessary condition $v \equiv 1 \bmod 96$ follows from (\ref{v.eq}). We prove sufficiency using the same approach as \cite{WS}.
Let $v = 96s+1$. For $s = 1,2$, $(96s+1, 3 \times 2,1)$-splitting BIBDs are given in 
\cite{WS} that are generated from base blocks over groups. Therefore, using Lemma \ref{group.lem}, we have
equitably ordered splitting BIBDs $s= 1,2$.

For $s\geq 3$, we proceed as follows.
A $3$-GDD of type $24^s$ exists for all $s \geq 3$. From Theorem \ref{general.thm}, 
we obtain a $(96s+1, 3 \times 4,1)$-splitting BIBD that is equitably ordered.
\end{proof}
In \cite[p.\ 674]{GMW}, Ge, Miao and Wang proved an asymptotic existence theorem for splitting BIBDs. We outline their approach now. First, they observed that a $(v, m \times c,1)$-splitting BIBD is equivalent to a decomposition of the complete graph $K_v$ into copies of $G$, where $G$ is the complete multipartite graph having $m$ parts of size $c$. Then the following result is an immediate of Wilson's theory of  ``graph designs.'' 

\begin{theorem}
\cite{GMW}
For fixed integers $m$ and $c$, there is an integer $v_{m,c}$ such that, for $v > v_{m,c}$, a $(v, m \times c,1)$-splitting BIBD exists if and only if  $v-1 \equiv 0 \bmod (m-1)c$ and  $v(v-1) \equiv 0 \bmod m(m-1)c^2$.
\end{theorem}

For equitably ordered splitting BIBDs, it is possible to use a recent extension of Wilson's theory due to Bowditch and Dukes \cite{BD} to obtain a similar asymptotic existence result. The paper \cite{BD} considers \emph{balanced graph decompositions} in which the graph $G$ is allowed to contain coloured loops. We start with the complete multipartite graph having $m$ parts of size $c$; however, we define $G$ by adding a loop having colour $i$ to every vertex in the $i$th part, for $1 \leq i \leq m$ (where we arbitrarily number the parts from $1$ to $m$). We also modify $K_v$ by placing $r/m$ loops of each of the $m$ colours at each vertex (where $r = (v-1)/ (c(m-1))$. It is not hard to see that a $G$-decomposition of the modified $K_v$ is equivalent to an equitably ordered  $(v, m \times c,1)$-splitting BIBD. Then \cite[Theorem 1.2]{BD} yields the following result.

\begin{theorem}
For fixed integers $m$ and $c$, there is an integer $v_{m,c}$ such that, for $v > v_{m,c}$, an equitably ordered $(v, m \times c,1)$-splitting BIBD exists if and only if  $v-1 \equiv 0 \bmod m(m-1)c^2$.
\end{theorem}

\section{Discussion and Conclusion}
Theorems~\ref{thm:devisacode} and \ref{thm:acodetoamd} show that a weak AMD code is in fact a special case of a group-generated authentication code for uniformly distributed sources.  The fact that these authentication codes have perfect secrecy gives a new perspective on the potential context in which an AMD code might be applied.  The traditional description of a weak AMD code involves an adversary who is unable to see an encoded message, but who can add a group element to that unknown message.  Thus an AMD code can only be applied in a context where these rather specific properties arise.  When treating the AMD code as an authentication code with perfect secrecy, it can be applied in any context where an authentication code might be useful.  Here the adversary sees the encoded message, but it is independent of the source and hence provides no information about the source. 

This perspective also allows us to identify those properties of an AMD code that do not hold for more general authentication codes.  For example, choosing a group element to add to the encoded message defines a substitution strategy for the authentication code.  We observe that each substitution strategy arising this way has the property that the probability of its success conditioned on the event of the key being $k$ is the same for all $k\in \mathcal{ K}$.  This property could facilitate an analysis of success probabilities of substitution attacks in a context where the adversary learns partial information about the choice of key, for example.

We have established that group-generated splitting set systems in general, and AMD codes in particular, are useful classes of splitting authentication codes with perfect secrecy.   It is interesting to see whether they can be further exploited in the construction of splitting authentication codes with perfect secrecy that achieve optimal or near-optimal security against substitution attacks, and whether an explicit focus on the perfect secrecy property can inspire new applications for AMD codes.


\section*{Acknowledgements} 
We thank Peter Dukes for bringing his paper \cite{BD} to our attention and for suggesting that it can be used to prove asymptotic existence of equitably ordered splitting BIBDs.

\end{document}